\documentclass[reqno,12pt]{amsart}

\usepackage{amsmath}
\usepackage{amssymb}
\usepackage{amsfonts}
\usepackage{amsthm}
\usepackage[foot]{amsaddr}
\usepackage{mathtools}
\usepackage{bbm}
\mathtoolsset{%
}
\usepackage{relsize}

\usepackage[utf8]{inputenc}
\usepackage[T1]{fontenc}

\usepackage{libertine}
\usepackage[libertine,libaltvw,cmintegrals
]{newtxmath}

\usepackage{dsfont}

\usepackage{mathrsfs}

\usepackage[%
cal=cm,
]
{mathalfa}

\usepackage{accents}

\usepackage[dvipsnames,svgnames]{xcolor}
\colorlet{MyBlue}{DodgerBlue!60!Black}
\colorlet{MyGreen}{DarkGreen!85!Black}

\usepackage{fullpage}
\newcommand{\afterhead}{.}
\usepackage[font=small,labelfont=bf]{caption}
\usepackage{subfigure}
\usepackage{tikz}
\usetikzlibrary{calc,patterns}
\usetikzlibrary{arrows,shapes,positioning}

\usepackage{acronym}
\usepackage{booktabs}       
\usepackage{latexsym}
\usepackage{paralist}
\usepackage{wasysym}
\usepackage{xspace}
\usepackage{comment}
\usepackage{longtable}
\usepackage[shortlabels]{enumitem}
\usepackage{pythontex}
\usepackage{listings}

\usepackage[authoryear,comma]{natbib}

\usepackage{hyperref}
\hypersetup{
colorlinks=true,
linktocpage=true,
pdfstartview=FitH,
breaklinks=true,
pdfpagemode=UseNone,
pageanchor=true,
pdfpagemode=UseOutlines,
plainpages=false,
bookmarksnumbered,
bookmarksopen=false,
bookmarksopenlevel=1,
hypertexnames=true,
pdfhighlight=/O,
urlcolor=MyBlue!60!black,linkcolor=MyBlue!70!black,citecolor=DarkGreen!70!black, 
pdftitle={},
pdfauthor={},
pdfsubject={},
pdfkeywords={},
pdfcreator={pdfLaTeX},
pdfproducer={LaTeX with hyperref}
}

\def\EMAIL#1{\email{\href{mailto:#1}{\texttt{\upshape #1}}}}

\numberwithin{equation}{section}  
\usepackage[sort&compress,capitalize,nameinlink]{cleveref}
\crefname{app}{Appendix}{Appendices}
\crefname{assumption}{Assumption}{Assumptions}

\crefrangeformat{equation}{\upshape(#3#1#4)\textendas\testh(#5#2#6)}

\DeclarePairedDelimiterX{\braket}[2]{\langle}{\rangle}{#1,#2}
\DeclarePairedDelimiterX{\inner}[2]{\langle}{\rangle}{#1,#2}
\DeclarePairedDelimiterX{\setdef}[2]{\{}{\}}{#1:#2}
\DeclarePairedDelimiterXPP{\probof}[1]{\prob}{(}{)}{}{%
#1}
\DeclarePairedDelimiterXPP{\exof}[1]{\ex}{[}{]}{}{%
#1}

\usepackage[textwidth=30mm]{todonotes}

\newcommand{\debug}[1]{#1}

\theoremstyle{plain}
\newtheorem{theorem}{Theorem}

\newtheorem*{corollary*}{Corollary}
\newtheorem{lemma}[theorem]{Lemma}

\theoremstyle{definition}
\newtheorem{definition}[theorem]{Definition}
\newtheorem*{definition*}{Definition}
\newtheorem{assumption}{Assumption}
\newtheorem*{assumption*}{Assumption}

\theoremstyle{remark}

\newtheorem*{remark*}{Remark}
\newtheorem*{notation*}{Notational remark}

\newtheorem{conjecture}[theorem]{Conjecture}
\newtheorem*{case*}{Case}

\DeclareMathOperator{\ex}{\debug{\mathbb{E}}}

\DeclareMathOperator{\prob}{\debug{\mathbb{P}}}

\setcitestyle{authoryear}

\title{Stochastic Consensus and the Shadow of Doubt}

\author{Emilien Macault$^{\dag,\ddag}$}
\address{$^{\dag}$ LEMMA, Université Paris 2 Panthéon-Assas, 16 rue Blaise Desgoffes, 75006 Paris, France}
\address{$^{\ddag}$ HEC Paris, 1 Rue de la Lib\'eration, 78350 Jouy-en-Josas, France}
\EMAIL{emilien.macault@hec.edu}

\begin{document}

\begin{abstract}
We propose a stochastic model of opinion exchange in networks. Consider a finite set of agents organized in a fixed network structure. There is a binary state of the world and, \textit{ex ante}, each agent is informed either about the true state of the world with probability $\alpha$ or about the wrong state with probability $1-\alpha$. We model beliefs as urns where white balls represent the true state and black balls the wrong state. Communication happens in discrete time and, at each period, agents draw and display one ball from their urn with replacement. Then, they reinforce their urns by adding balls of the colors drawn by their neighbors. We show that this process converges almost-surely to a stable state where all urns have the same proportion of balls. We show that this limit proportion is a random variable with full support over $\left[0,1\right]$. We propose a conjecture on the distribution of this limit proportion based on simulations.

\smallskip


\end{abstract}
\maketitle

\newpage

\section{Introduction}

Despite public investments in media education and the development of counter-measures over the past few years, misinformation remains an ongoing issue with tangible consequences. The recent examples of the COVID-19 pandemics or the US presidential elections have highlighted how quickly inaccurate, deceptive or politically biased information spreads in a context of distrust towards experts and institutions.

Under pressure to limit the spread of deceptive content, media and open web companies have put in place a set of policies to regulate news contents on their platforms. Such policies mostly include source highlighting, fact checking and advertisement campaigns, all of which have proved to have a limited efficiency. The failure of counter-disinformation policies may be explained by theoretical modeling shortcomings. Indeed, most policies are based on the assumption that agents behave rationally when it comes to information processing. By displaying the limited trustworthiness level of a spurious source, they assume agents will revise their beliefs over secure sources and naturally evacuate false news. In terms of economic modeling, this translates into the assumption that agents behave in a Bayesian manner. In a multi-agent context, the computational limitations of Bayesian models have incentivized the emergence of non-Bayesian models as an alternative. Most of these models are based on \cite{DG74} and consist in setups where agents communicate by repeatedly averaging their opinions with their neighbors' until a steady state is reached. The strength of this framework is that beliefs quickly converge to a tractable limit.

Yet, two major limitations are to be opposed to models based on DeGroot dynamics. First, they have been shown to have limited robustness, in the sense that the repeated averaging overweights initial beliefs while enforcing fast convergence of beliefs. Second, one may question the relevance of a setup where agents directly access and exchange their beliefs on some state. In most small-world communication setups, it seems more realistic to assume that agents do not access the subjective probabilities they put on the possible values of a state, but rather decide to relay some information over another according to the relative probabilities they put on those events. In other terms, they do not exchange beliefs but draws based on beliefs. In this paper, we introduce and analyze a stochastic variant of DeGroot dynamics where agents behave in this respect.

Introducing some degree of stochasticity strongly changes the perspective on misinformation: in \cite{DG74}, slightly modifying the prior beliefs of some agents cannot change drastically the consensus outcome. Yet, most disinformation platforms display some "shadow of doubt" strategy: agents do not transmit false informations because they necessarily believe them to be true, but rather because there is some -- even small -- probability that they may not be false. In other words, they manage to disinform by inducing limited beliefs on their false informations.

\subsection{Contribution}

In this paper, we build an opinion formation model where agents communicate by drawing states according to their beliefs instead of directly communicating subjective probabilities. To do so, we model beliefs using reinforcing urns. Studying the evolution of beliefs comes down to characterizing the evolution of urns' compositions. Using stochastic approximation techniques, we show that in such models, under very general conditions, the dynamics of beliefs converge to a rest point. We show that at the steady-state, all agents share the exact same belief on the state of the world. As long as initial beliefs cover the whole state space, the consensus is drawn from a distribution with full support. This strongly contradicts the predictions of DeGroot's and similar models. We then try to characterize this limit distribution using simulations.

\subsection{Literature}




The emergence of consensus and its connection to learning heuristics is a long lasting question in theoretical economics. Two approaches are generally opposed to this problem: Bayesian and non-Bayesian models.

Bayesian models on the emergence of consensus mostly started with \cite{A76} and its seminal result that two agents with equal prior beliefs and common knowledge posteriors must have equal beliefs. Generalizations have been proposed by \cite{GP82} and \cite{PK90} who showed respectively that two players repeatedly communicating must agree in the long run and that a finite umber of players communicating in pairs will eventually agree. This is partly due to actions being observable, although \cite{BG98} proves that when agents are embedded in a connected network and observe the outcome of their actions with some noise, players are able to learn the true payoff distributions for actions that their neighbors take infinitely often, hence all actions converge to a consensual action. It is worth noting that although \cite{BG98}considers Bayesian agents, they limit their ability to compute beliefs by assuming that they do not make inferences on unobserved players and behave myopically. These results are closely related to social learning models and the observational learning literature as \cite{B92}, \cite{SS00} or \cite{RSV09}. \cite{MMST20} generalizes those results to a large class of social learning models by introducing the concept of \emph{social learning equilibrium} to study the asymptotic properties of learning processes and characterize conditions that agreement and herding behavior. In the same line, \cite{ADLO11} and \cite{AOB14} connect the emergence of social learning with Bayesian agents and the topology of the communication network.

Non-Bayesian consensus models emerged through \cite{DG74}, where the author introduces a model where agents living in a network repeatedly exchange their beliefs over some state of the world. At each stage, each agent replaces his belief by the average of his neighbors' beliefs. It is shown that if the communication network is connected, beliefs converge to a consensus which depends on initial beliefs and the network topology only. Variations on the updating rule have been proposed, for instance in \cite{FJ97} where authors allow some persistence on agents beliefs by including one's own belief in the averaging process. DeGroot's model gained popularity with \cite{GJ10} which connects consensus and learning with the network's adjacency matrix using properties of Markov chains steady-states. They refer to DeGroot's belief averaging dynamics as \emph{naive learning}. Examples of the use of \cite{DG74} in economic modeling are too numerous to be listed. In the recent literature, \cite{MV20} considers a stochastic game where two misinformers try to influence a population of agents applying naive learning. We refer the reader to \cite{AO11} and \cite{GS17} for surveys on both Bayesian and non-Bayesian learning in networks.

In this paper, we question the limits of DeGroot's model by allowing agents to communicate \textit{via} draws made according to their beliefs. In this regard, this work is in line with the existing literature on robustness of learning dynamics. \cite{ACY16} questions the predictions of Bayesian learning models by introducing uncertainty on the distribution of private signals agents receive. Criticism on DeGroot dynamics already featured in \cite{GJ10} where authors proved that in the general case, beliefs do not converge for countably infinite player sets. In a recent work, \cite{PAAA21} shows that in the presence of agents with fixed beliefs over time, which they call \emph{bots}, the common limit can converge to any value.

Our approach to this problem is based on the seminal Polya urn model from \cite{EP23}. In this paper, authors consider an urn with balls of several colors and study the convergence of reinforcement dynamics. It is a well known result that proportions in the urn converge to a beta distribution (see \cite{K13} for instance). The strength of Polya's model is its intricate connection with exchangeability. Central papers in the foundation of Bayesian inference like \cite{DF29} and \cite{HS55} heavily rely on the concept of exchangeability. Polya urn plays a particular role in that \cite{HLS87} proved that any exchangeable process of $\left\{0,1\right\}$-valued random variables if either Bernoulli, deterministic or generated by draws from a Polya urn. Our model considers a system of interacting urns in the flavor of \cite{PS04} which introduced such systems and first proved convergence when the number of balls in all urns grow at the same speed. Similarly, \cite{DLM14} shows convergence of proportions in a system where urn reinforcement depends both on proportions in each urn and on the average proportions in the system. \cite{CDM16} gives further results on convergence and fluctuations around the limit of such system. The model we consider is close yet different, as we consider a system where urns are reinforced at different speeds which correspond to their degree in the communication network. Usual probabilistic tools do not apply as proportions are not martingales and draws are not exchangeable. Instead, we rely on stochastic approximation as introduced by \cite{RM51}. Motivation for the use of stochastic approximation in the study of urn systems can be found in \cite{LP13}, where authors use this technique in the context of clinical trial modelling.

\section{The Model}

\subsection{Model} 

We consider a finite population $V=\left\{1,\dots N\right\}$ of $N$ agents, embedded in a exogeneous and fixed undirected graph $G=(V,E)$ with edge set $E$. We denote by $N(i)$ the neighborhood of any player $i$ and $d_i$ the degree of $i$. We denote by $A=(a_{ij})_{i,j \in N}$ the adjacency matrix of $G$, with the convention that $a_{ii}=0$ for all $i \in V$. Throughout, we will assume that $G$ is \emph{connected} that is, for every pair of nodes $i,j$ in $E$ there exists a path in $G$ connecting $i$ to $j$.

We consider a binary state space $\Theta=\{g,b\}$. At the beginning of the game, a state $\theta$ is drawn at random from $\Theta$ and is unobserved. Players initially receive a noisy signal informing them about the state of the world. Formally, with $\tilde{\theta}$ being the realized value of the state, player $i$ receives a signal 

$$\begin{cases} \theta=\tilde{\theta} \text{ with probability } \alpha \\ \theta=\tilde{\theta}^c \text{ with probability } 1-\alpha \end{cases}$$

With $\tilde{\theta}^c$ representing the complementary value of $\tilde{\theta}$ in $\Theta$. For a large network, by the law of large numbers, $\alpha$ represents the average proportion of agents initially well-informed.

\subsection{Beliefs and communication}

We model beliefs using urns of infinite capacity with balls of two colors representing the possible values of the state $\theta$: white balls represent the event $\theta=\tilde{\theta}$ and black balls represent the complementary event. At any time $t\ge 1$, the proportions of balls in agent $i$'s urn then represent $i$'s belief over those two events. Urns are initialized with a ball corresponding to the agent's signal.

We define the communication process as a discrete time dynamics. At each stage, players draw one ball with replacement from their urns with uniform probability. Draws are assumed to be pairwise independent. Every agent observes the colors drawn by their neighbors. Then, beliefs are updated by reinforcing urns, adding one ball of the corresponding color for each draw from their neighbors. The process is repeated infinitely.

At any given time, an agent's current belief on $\theta$ is given by the proportions of balls in his urn. Our objective is to study the evolution of the urn system and determine whether proportions converge, if a consensus is reached and if so, to characterize it given the network topology and the value of $\alpha$. 

\subsection{Example}

Consider $N=3$ agents connected in line as displayed in \cref{fi:ex}. Assume that at time $t=0$, agent 1 and 3 received a truthful signal and agent 2 got the wrong one. Then urns 1 and 3 will contain a white ball and urn 2 will contain one blackball. At time $t=1$, every player draws the only ball their urns contain and display it. Then, they all replace their draw and add a new ball of the corresponding color for every draw they observe. That is, at the end of the first stage, urns 1 and 3 will contain each one black ball and one white ball, and urn 2 will contain one black ball and two white balls. 

\begin{figure}[h]
\centering
\begin{tikzpicture}[thick]
	\tikzset{MidNode/.style = {draw, circle}}
    \node[MidNode](a) at (-2,0) {$1$};
		\node[MidNode](b) at (0,0){$2$};
		\node[MidNode](c) at (2,0){$3$};
		\draw[-](a)--(b);
		\draw[-](b)--(c);
\end{tikzpicture}
\caption{\label{fi:ex} Three urns in line.}
\end{figure}
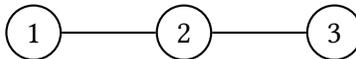

At time $t=2$, the draw and reinforcement procedure is repeated: the end urns will draw one black or one white ball with equal probability and the middle urn will draw one white ball with probability $2/3$ and a black ball with probability $1/3$.

\begin{table}[h]
\centering
  \begin{tabular}{ c | c | l }
		Vector of draws & Probability of occurence & Urns compositions \\
    \hline
    (W,W,W) & 1/6 & (2,1)--(4,1)--(2,1) \\ \hline
    (W,B,W) & 1/12 & (1,2)--(4,1)--(1,2) \\ \hline
		(B,W,W) & 1/6 & (2,1)--(3,2)--(2,1) \\ \hline
		(B,B,W) & 1/12 & (1,2)--(3,2)--(1,2) \\ \hline
		(W,W,B) & 1/6 & (2,1)--(3,2)--(2,1) \\ \hline
		(W,B,B) & 1/12 & (1,2)--(3,2)--(1,2) \\ \hline
		(B,W,B) & 1/6 & (2,1)--(1,4)--(2,1) \\ \hline
		(B,B,B) & 1/12 & (1,2)--(1,4)--(1,2) \\ \hline
    \hline
  \end{tabular}
\caption{Possible outcomes at time $t=2$.}
\label{tab:ex}
\end{table}

\cref{tab:ex} details the possible outcomes at time $t=2$. The left column is the vector of draws from urn 1, 2 and 3 respectively and the right column gives the compositions at the end of the time period in the same order. Left figures correspond to the number of white balls in the urn and right figures to the number of black balls.

\section{Results}

\subsection{Urn dynamics}

To ease the characterization of the dynamics, we introduce some notations. First, denote by $B_i^t$ and $W_i^t$ respectively the numbers of white and black balls in player $i$'s urn at time $t$. Define $S_i^t=B_i^t+W_i^t$ the total number of balls in player $i$'s urn at time $t$. One has that $S_i^{t+1}=S_i^t+d_i$ and $S_i^0=1$ hence $S_i^{t}1+d_i t$. Let $Z_i^t=B_i^t/S_i^t$ be the proportion of black balls in urn $i$ after step $t$ and $Z^t=(Z_i^t)_{i \in V}$. Finally, let $X_i^t$ be the indicator variable of a black draw for agent $i$ at time $t$, $X^t=(X_1^t,\dots,X_N^t)$ and let $\mathcal{F}_t$ be the sigma-field generated by the realizations of $(X^k), k\leq t$.

We derive the following dynamics:

\begin{equation}
\label{eq:dyn}
\begin{cases}
	B_i^{t+1}=B_i^t+\sum_{j \in N(i)}X_j^{t+1}\\
	W_i^{t+1}=W_i^t+d_i-\sum_{j \in N(i)}X_j^{t+1}\\
	S_i^{t+1}=S_i^t+d_i=1+d_i(t+1)
\end{cases}
\end{equation}

Hence 

\begin{equation}
\label{eq:expected_dyn}
\mathbb{E}\left[\Delta B_i^t|\mathcal{F}_t\right]=\mathbb{E}\left[\sum_{j \in N(i)} X_j^{t+1}| \mathcal{F}_t\right]=\sum_{j \in N(i)} Z_j^t
\end{equation}

\cref{eq:expected_dyn} shows that, in expectation, the belief updating process obeys some local averaging property as in canonical naive learning models: the variation of proportion in any urn evolves according to the proportions in the neighboring urns.

\subsection{Convergence of beliefs}

We first show that beliefs converge in the sense that color proportions in each urn converge to a stable point. The proof relies on stochastic approximation techniques, as usual probability methods do not apply in our case. Indeed, unless the graph $G$ is regular, neither local nor global proportions behave as martingales and it is easy to see that the process $(Z^t)$ is not exchangeable as the rate at which an urns evolve depends both its degree and time. Nevertheless, we are able to frame the dynamics as an algorithm for which we can prove convergence.

\begin{theorem}
\label{th:conv}
For any graph $G$, $\lim_{t\rightarrow \infty}Z_i^t$ exists almost-surely for any $i\in V$.
\end{theorem}

\begin{proof}

From \cref{eq:dyn} we derive the following recursive formula on $Z_i^t$:

\begin{equation}
Z_i^{t+1}-Z_i^t=\cfrac{-d_i Z_i^t + \sum_{j \in N(i)} X_j^{t+1}}{1+d_i(t+1)}
\end{equation}

By adding and subtracting the conditional expectation of the number of black draws in neighboring urns to the numerator, we have:

\begin{equation}
Z_i^{t+1}-Z_i^t=\cfrac{-d_i Z_i^t + \mathbb{E}\left[\sum_{j \in N(i)} X_j^{t+1}| \mathcal{F}_t\right] + \sum_{j \in N(i)} X_j^{t+1} - \mathbb{E}\left[\sum_{j \in N(i)} X_j^{t+1}| \mathcal{F}_t\right]}{1+d_i (t+1)}
\end{equation}

Observing that, conditional on $\mathcal{F}_t$, the expected number of black draws in neighboring urns at time $t+1$ is equal to the sum of their proportions at time $t$, we obtain:

\begin{equation}
Z_i^{t+1}-Z_i^t=\cfrac{-d_i Z_i^t + \sum_{j \in N(i)} Z_j^{t} + \sum_{j \in N(i)} X_j^{t+1} - \mathbb{E}\left[\sum_{j \in N(i)} X_j^{t+1}| \mathcal{F}_t\right]}{1+d_i (t+1)}
\end{equation}

We now rescale the equation by a factor that is independent of $d_i$:

\begin{equation}
Z_i^{t+1}-Z_i^t=\cfrac{\cfrac{1+d(t+1)}{1+d_i (t+1)}\left(-d_i Z_i^t + \sum_{j \in N(i)} Z_j^{t} + \sum_{j \in N(i)} X_j^{t+1} - \mathbb{E}\left[\sum_{j \in N(i)} X_j^{t+1}| \mathcal{F}_t\right]\right)}{1+d(t+1)}
\end{equation}

Where $d=\min_i d_i$. Finally, we obtain the following system:

\begin{equation}
\label{eq:system}
Z^{t+1}-Z^t=\gamma^t\left[f^t(Z^t)+u^t\right]
\end{equation}

Where 

\begin{equation}
\begin{cases}
\gamma^t=\cfrac{1}{1+d(t+1)},\\ 
 \\
f^t:\left[0,1\right]^N \to \left[0,1\right]^N\\
\qquad\quad Z^t \mapsto (f_1^t(Z^t),\dots,f_N^t(Z^t))\\
 \\
\text{with } f_i^t(Z^t)=\cfrac{1+dt}{1+d_i t}\sum_{j \in N(i)} Z_j^t - d_i Z_i^t,\\
 \\
u^t=(u_1^t,\dots,u_N^t),\\
\text{with } u_i^t=\cfrac{1+dt}{1+d_i t}\sum_{j \in N(i)} X_j^{t+1} - \mathbb{E}\left[\sum_{j \in N(i)} X_j^{t+1}| \mathcal{F}_t\right]
\end{cases}
\end{equation}

\medskip

In order to ensure convergence of the stochastic system \cref{eq:system}, we first make the following observations.

\medskip

\begin{assumption}
\label{obs:step}
\begin{equation}
\begin{cases}
\sum_{t=1}^{\infty} \gamma^t = \infty \\
\sum_{t=1}^{\infty} \left(\gamma^t\right)^2 < \infty
\end{cases}
\end{equation}
\end{assumption}

\cref{obs:step} is central in any stochastic approximation algorithm \textit{à la} \cite{RM51} with deterministic weights. These weights serve as the increments of time discretization. In that perspective, the first point implies that the algorithm will cover the entire time interval. The second point involves, jointly with the next observation, the disappearing of noise in the limit. As $\gamma^t$ is of the order of $1/t$, \cref{obs:step} is immediate.

\bigskip

\begin{assumption}
\label{obs:noise}
For every $i$ in $V$, the sequence $(u_i^t)$ is a martingale difference noise relative to $\mathcal{F}_t$.
\end{assumption}

\cref{obs:noise}, when combined with the second point in \cref{obs:step}, ensures that the cumulative error due to the discretization noise is negligible almost-surely, as the noise variance will vanish asymptotically. \cref{obs:noise} holds as, for any $i\in V$, the sequence $(u_i^t)$ is a sequence of bounded random variables with zero mean.

\bigskip

\begin{assumption}
\label{obs:f}
The maps $f_i^t$ are Lipschitz continuous and measurable with respect to $\mathcal{F}_t$ and uniformly continuous in $t$ for $t\ge 1$.
\end{assumption}

Stochastic approximation ensures that a discrete-time stochastic process evolves along the trajectories of a continuous time ordinary differential equation. In that respect, \cref{obs:f} ensures that the ODE is well defined and has a unique solution.
 
Finally, although the maps $(f_i^t)_{i,t}$ in \cref{eq:system} depend on time, for any $i\in V$, the sequence of maps $(f_i^t)$ converge to a time-independent limit as time goes to infinity. Indeed, for any $i \in V$ and any $z\in\left[0,1\right]^N$, let $\bar{f}_i(z)=\frac{d}{d_i}\sum_{j \in N(i)} z_j - d_i z_i$ and $\bar{f}:z \mapsto \left(\bar{f}_1(z), \dots, \bar{f}_N(z)\right)$.

\begin{assumption}
\label{obs:cont}
For any $z \in \left[0,1\right]^N$ and any $k\in \mathbb{N}^*$,

$$
\lim_{s \to \infty} \left|\sum_{t=s}^{s+k} \gamma^t\left[f_i^t(z)-\bar{f}_i(z)\right]\right| \to 0
$$
\end{assumption}

\cref{obs:cont} holds immediately as, for any $i \in V$, $f_i^t \to \bar{f}_i$ as $t \to \infty$.

Based on \crefrange{obs:step}{obs:cont}, we can apply Theorem 2.3 from \protect{\citet[chapter 5]{KY03}}.

\begin{theorem}[\protect{\citet[chapter 5]{KY03}}]
If \crefrange{obs:step}{obs:cont} hold and $\left(Z^t\right)$ is bounded with probability one, then for almost all $\omega$, the limits $\bar{Z}(\omega)$ of convergent subsequences of $\left(Z^t(\omega)\right)$ are trajectories of 

\begin{equation}
\label{eq:ode}
\dot{z_i^t}=\bar{f}(z^t)
\end{equation}

in some bounded invariant set and $\left(Z^t(\omega)\right)$ converges to this invariant set.

\end{theorem}

This result ensures that the system \cref{eq:system} evolves almost-surely along trajectories of \cref{eq:ode} and converges to the set of asymptotically stable points of the ordinary differential system.

\end{proof}

This first result ensures that for any graph structure $G$ and any initial condition on the urns, proportions converge almost-surely to a stable point. In particular, convergence is independent of the initial signal structure and applies for any alternative initialization of the system. The next result details when a consensus emerges.

\subsection{Emergence of Consensus}

\begin{theorem}
\label{th:consensus}
Suppose that the graph $G$ is connected. Then for any $i,j \in V$, $\lim_{t\rightarrow \infty}Z_i^t=\lim_{t\rightarrow \infty}Z_j^t$ almost-surely.
\end{theorem}

\begin{proof}

From \cref{th:conv}, we know that proportions converge along the trajectories of \cref{eq:ode}, that is:

\begin{equation}
\dot{z_i^t}=\cfrac{d}{d_i}\sum_{j \in N(i)} z_j^t - d_i z_i^t
\end{equation}

As $\frac{1}{N}\le \frac{d}{d_i} \le 1$, stable points of \cref{eq:ode} belong to the set of stable points of 

\begin{equation}
\dot{z_i^t}=\sum_{j \in N(i)} z_j^t - d_i z_i^t
\end{equation}

i.e. 

\begin{equation}
\dot{z}^t=-Lz^t
\end{equation}

Where $L$ is the Laplacian matrix of the graph $G$, i.e. $L=D-A$ with $D$ the diagonal matrix of degrees. 

Thus, Lyapounov stable solutions of \cref{eq:ode} belong to the nullspace of $-L$, as $L$ is symmetric, positive semi-definite. As $G$ is connected, this nullspace is of dimension 1 and is characterized by the eigenvector $(1,...,1)$ as the sum of each row in $L$ equals zero. As $-L$ is negative, the entire set is Lyapounov stable.\\


\end{proof}

\begin{theorem}
\label{th:conv}
For any connected graph $G$, if $\alpha\in\left(0,1\right)$, then the limit belief $\bar{Z}$ is a non-trivial distribution with full support on $\left[0,1\right]$.
\end{theorem}

\begin{proof}

The proof is based on the concept of \emph{attainability} from \cite{B99}.

\begin{definition}
A point $p\in \mathbb{R}^N$ is \emph{attainable} by $Z$ if for every $t >0$ and every open neighborhood $U$ of $p$, 

$$
\mathbb{P}\left(\exists s \ge t: Z^s \in U\right)>0.
$$
\end{definition} 

In other terms, a point $p$ is attainable if, from any vector of proportions, there is a strictly positive probability that $Z^t$ becomes arbitrarily close to $p$ in finite time. Let $L_{\bar{f}}$ denote the set of equilibrium points of \cref{eq:ode} intersected with $\left[0,1\right]^N$ that is, $L_{\bar{f}}=\left\{z\in\left[0,1\right]| z_i=z_j \forall i,j \le N\right\}$. We establish the following lemma.

\begin{lemma}
\label{le:att}
Any point $p$ in $L_{\bar{f}}$ is attainable.
\end{lemma}

To prove \cref{le:att}, simply observe that, from \cref{eq:dyn}, $Z_i^{t+1}-Z_i^t$ is of the order of $\frac{1}{t}$. If $\alpha\in\left(0,1\right)$, $\mathbb{P}\left(Z_i^t \in \left(0,1\right)\right)>0$ for every $i$ and $t\ge 0$.

We showed that any point in $L_{\bar{f}}$ is attainable. To complete the proof of \cref{th:conv}, it remains to show that any attainable point in $L_{\bar{f}}$ belongs to the support of $Z^{\infty}$.

\end{proof}

Observe that if $\alpha=0$ or $\alpha=1$, urns in the system display only one color hence beliefs will remain at their original value forever.

The next section provides some empirical evidence on the limit distribution of beliefs.

\section{Limit Distribution}

While our efforts in characterizing the limit distribution of the consensus $Z$ as a function of $G$ and $\alpha$ failed, large scale simulations provide some useful evidence. We simulated the learning dynamics on three network structures: stars, regular graphs with varying degree and complete networks. The values of the limit belief were simulated for different values of $\alpha$. Main elements of code used for the simulation feature in the appendix.

\subsection{Evidence of a Beta Distribution}

In the classical model from \cite{EP23}, an urn is initialized at time $t=0$ with $\alpha \ge 1$ white balls and $\beta \ge 1$ black balls. Then, at each discrete time step, a ball is drawn from the urn and replaced with $m\ge 1$ additional balls of the same color. It is widespread that the proportion of white balls converges in distribution to a beta distribution $\mathcal{B}(\alpha /m,\beta /m)$ (see \cite{M08}). 

For any two reals $a,b>0$, the beta distribution $\mathcal{B}(a,b)$ has a density function

\begin{equation}
p(x,a,b)=\cfrac{x^{a-1}(1-x)^{b-1}}{B(a,b)} \mathbbm{1}_{\left\{x\in\left[0,1\right]\right\}}
\end{equation}

where $B(a,b)=\frac{\Gamma(a)\Gamma(b)}{\Gamma(a+b)}$ and $\Gamma$ is the Gamma function.

Although we consider a system of interacting urns rather than a single urn, the beta distribution stands as a strong candidate for the limit distribution.

\begin{figure}[h]
\includegraphics[width=\columnwidth]{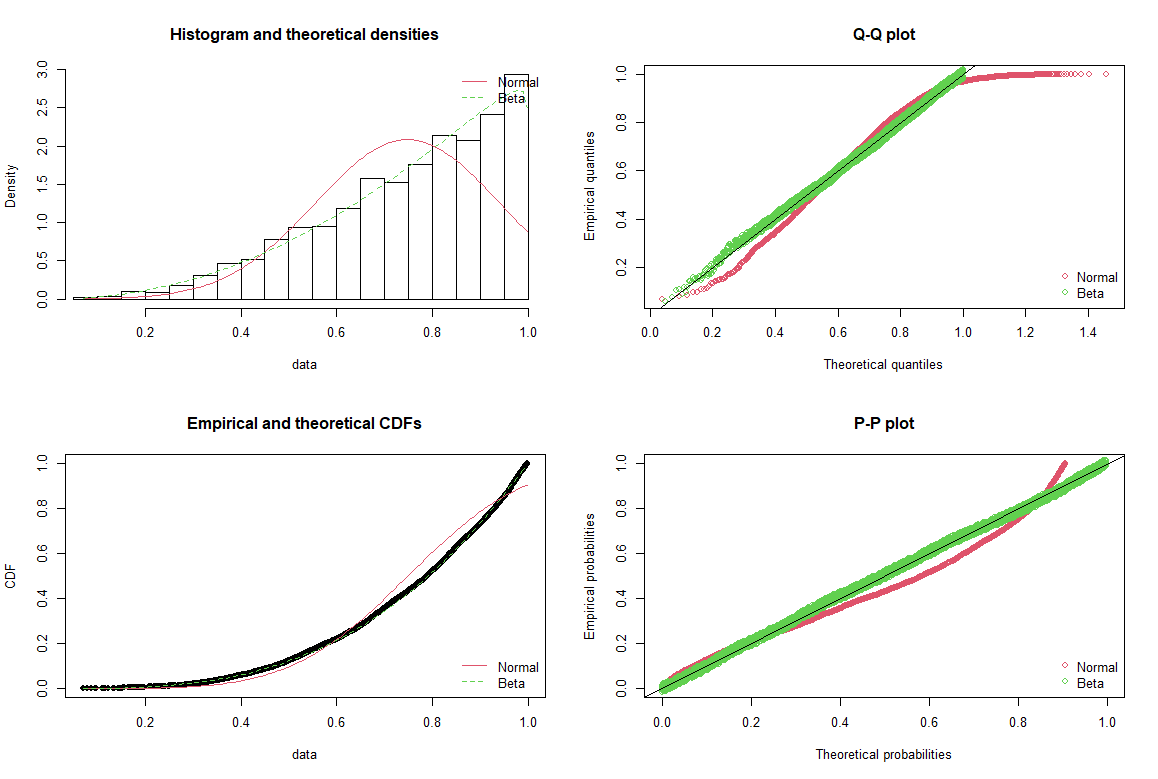}
\caption{Fitness measures for beta and normal distributions on a star graph with $\alpha=0.75$, $n=5000$ observations.}
\label{fig:betafit}
\end{figure}

\begin{conjecture}
\label{conj:dist}
The distribution of $\bar{Z}$ follows a beta distribution $\mathcal{B}(a,b)$ for some $a,b,>0$ which depend only on $\alpha$ and $G$.
\end{conjecture}

We collected the values of the limit belief for star, $k$-regular and complete graphs of fixed size $N=100$ and fixed values of $\alpha$. As all the simulations were run independently, for any given graph structure and value of $\alpha$, the set of values of the limit beliefs is an i.i.d. sample. A beta distribution fitting was computed by maximum likelihood estimation. Other distributions were fitted in order to assess goodness-of-fit using usual criteria. \cref{fig:betafit} compares the fitted distributions assuming respectively a normal distribution and a beta distribution. The graph displays empirical and theoretical densities, quantile-quantile plots, cumulative distribution functions and probability plots. The beta distribution clearly appears as well fitted to the sample. Additional plots feature in the appendix for different network structures and values of $\alpha$. In all the aforementioned cases, fitness measures yielded similar results, where the beta distribution clearly appears as more adapted to describe the data.

\subsection{Estimation of the Parameters}

Assuming the limit belief does follow a beta distribution, we are able to estimate its parameters using maximum likelihood estimation for various networks and values of $\alpha$. Our first conjecture concerns the average of the limit distribution. All the simulations we ran conducted to a strong belief in that its value is $\alpha$. In other terms, the expected proportion of black balls in the limit is equal to the expected number of misinformed agents \textit{ex-ante}.  

\begin{conjecture}
\label{conj:mean}
For any $i \in N$, $Z_i^t \to \bar{Z}$ where $\bar{Z}\sim\mathcal{B}(a,b)$ with $a,b>0$ such that $\frac{a}{a+b}=\alpha$.
\end{conjecture}

To support this conjecture, we simulated the communication dynamics on a $10$-regular network for increasing values of alpha and fitted a beta distribution to the empirical distribution. \cref{tab:est} provides estimates of the parameters and sample means.

\begin{table}[h]
\centering
  \begin{tabular}{ c | c | c | c | c }
		$\alpha$ &  $\hat{a}$ & $\hat{b}$ & $\hat{a}+\hat{b}$ & Empirical mean\\
    \hline
		0.1 & 6.09 & 54.63 & 60.72 & 0.100 \\
		0.2 & 11.9 & 47.6 & 59.5 & 0.199 \\
		0.3 & 18.3 & 42.7 & 61 & 0.300 \\
		0.4 & 24.8 & 37.2 & 62 & 0.399 \\
		0.5 & 30.8 & 30.8 & 61.6 & 0.499 \\
		0.6 & 37.4 & 24.9 & 62.3 & 0.601 \\
		0.7 & 43.1 & 18.5 & 61.5 & 0.699 \\
		0.8 & 48.4 & 12.1 & 60.5 & 0.8005 \\
		0.9 & 51.19 & 5.68 & 56.87 & 0.900 \\
    \hline
  \end{tabular}
\caption{Estimates $\hat{a}$ and $\hat{b}$ of $a$ and $b$ and empirical mean in a $10$-regular network of size $N=100$ (14000 obs.).}
\label{tab:est}
\end{table}

Similar simulations for other network structures yield identical results. These outputs strongly suggest that the belief updating dynamics\cref{eq:dyn} bear some form of asymptotical exchangeability. It remains to be proved yet it would strongly support \cref{conj:dist}. Finally, we observe consistency in the sum of the estimates $\hat{a}$ and $\hat{b}$, which support the following conjecture.

\begin{conjecture}
Fix a network $G$ and let $\bar{Z}_{\alpha}$ be the beta distribution of the limit belief given $\alpha$. Let $a_{\alpha}$ and $b_{\alpha}$ be its parameters. Then the mapping $\alpha \mapsto a_{\alpha} + b_{\alpha}$ is constant.
\end{conjecture}

To support this conjecture, we explored results from simulations on the three aforementioned structures. We believe that discrepancies as observed in \cref{tab:est} are due to noise introduced by the random number generator as they mostly appear for extreme values of $\alpha$.

\section{Future Work}

Further work has been done in trying to characterize the limit distribution of the consensus analytically, yet no method has yield convincing results so far. Based on the simulations we ran, we have strong hints that suggest this limit distribution is a beta distribution. This confirms the intuition that, as a whole, the system acts as a global Polya urn. Further exploration of possible exchangeability properties might help in supporting this intuition and provide some tools for a closed-form characterization. To our knowledge, no paper has been able to achieve such formal results on interacting urn systems. \\

A better understanding of the limit distribution would provide a better applicability of our results, in particular in designing a model of strategic disinformation with disinformants being parts of the network. That application was the initial motivation of the paper and remains its main objective.

\bibliographystyle{apalike}
\bibliography{proj_nx}

\end{document}